\documentclass{article}
\usepackage[utf8]{inputenc}
\usepackage{amsmath}
\usepackage{amssymb}
\usepackage{amsthm}
\usepackage{fullpage}
\usepackage{enumitem}
\usepackage[dvipsnames,table]{xcolor}
\usepackage{caption}
\usepackage{subcaption}
\usepackage{floatrow}
\usepackage{biblatex} %Imports biblatex package
\usepackage{multirow}
\usepackage{floatpag}
\usepackage{wrapfig}

\usepackage{euler} % police maths
\usepackage{palatino} % police
\usepackage{fourier-orns} % leaf
\usepackage{pifont} % xmark
\usepackage[colored]{shadethm}

\usepackage{pifont}% http://ctan.org/pkg/pifont
\newcommand{\xmark}{\ding{55}}%

\definecolor{shaderulecolor}{rgb}{1,1,1} % couleur de l'encadré

\colorlet{shadethmcolor}{green!80!black!10}
\colorlet{GreenBlack}{green!60!black!80}

\newshadetheorem{theorem}{Theorem}
\newshadetheorem{proposition}{Proposition}
\newshadetheorem{lemma}{Lemma}
\newshadetheorem{corollary}{Corollary}
\newshadetheorem{remark}{Remark}
\newshadetheorem{definition}{Definition}
\newshadetheorem{assumption}{Assumption}

 \renewenvironment{proof}{{\noindent\bfseries Proof.}}{\hfill{\color{GreenBlack}{\leafNE}}}
 
\usepackage{hyperref}
\hypersetup{
colorlinks=true,
breaklinks=true, 
urlcolor= GreenBlack, 
linkcolor= GreenBlack,
citecolor=GreenBlack,
}

\colorlet{lblue}{blue!50!white}
\colorlet{lred}{red!50!white}
\colorlet{lgreen}{green!50!white}
\colorlet{lpurple}{purple!50!white}
\colorlet{lorange}{orange!50!white}
\colorlet{lpink}{pink!50!white}
\colorlet{lbrown}{brown!50!white}
\colorlet{lyellow}{yellow!50!white}
\colorlet{lolive}{olive!50!white}

\usepackage[color, leftbars]{changebar}
\makeatletter
\let\@footnotetext\ltx@footnotetext
\makeatother

\usepackage{graphicx}
\usepackage{bbold}
\usepackage{arydshln}

\usepackage[ruled,vlined,linesnumbered]{algorithm2e}
\DontPrintSemicolon
\usepackage{mathtools}

\usepackage{tikz}
\usetikzlibrary{decorations.pathreplacing}
\newcommand{\node}[1][white]{
\begin{tikzpicture}
\tikzstyle{noeud}=[draw,circle,fill=white,scale=0.7]
\node[noeud,fill=#1] (0) at (0,0) {};
\end{tikzpicture}
}

\usetikzlibrary{decorations.markings}
\tikzset{->-/.style={decoration={
  markings,
  mark=at position .5 with {\arrow{>}}},postaction={decorate}}}

\DeclareMathOperator{\parent}{\mathcal{P}}
\DeclareMathOperator{\children}{\mathcal{C}}
\DeclareMathOperator{\depth}{\mathcal{D}}
\DeclareMathOperator{\leaves}{\mathcal{L}}

\DeclareMathOperator{\multiset}{\children_c}
\DeclareMathOperator{\length}{length}
\DeclareMathOperator{\width}{\mathcal{W}}
\DeclareMathOperator{\red}{\mathfrak{R}}

\title{Revisiting Tree Isomorphism:\\ An Algorithmic Bric-à-Brac}

\author{Florian Ingels\\
\url{florian.ingels@univ-lille.fr}}

\date{Univ. Lille, CNRS, Centrale Lille, UMR 9189 CRIStAL, F-59000 Lille, France}

\begin{document}

\maketitle

\begin{abstract}
\noindent
The Aho, Hopcroft and Ullman (AHU) algorithm has been the state of the art since the 1970s for determining in linear time whether two unordered rooted trees are isomorphic or not. However, it has been criticized (by Campbell and Radford) for the way it is written, which requires several (re)readings to be understood, and does not facilitate its analysis. In this article, we propose a different, more intuitive formulation of the algorithm, as well as three propositions of implementation, two using sorting algorithms and one using prime multiplication. Although none of these three variants admits linear complexity, we show that in practice two of them are competitive with the original algorithm, while being straightforward to implement. Surprisingly, the algorithm that uses multiplications of prime numbers (which are also be generated during the execution) is competitive with the fastest variants using sorts, despite having a worst theoretical complexity.  We also adapt our formulation of AHU to tackle to compression of trees in directed acyclic graphs (DAGs). This algorithm is also available in three versions, two with sorting and one with prime number multiplication. Our experiments are carried out on trees of size at most $10^6$, consistent with the actual datasets we are aware of, and done in Python with the library \texttt{treex}, dedicated to tree algorithms.

\medskip

\end{abstract}

\section{Introduction}

\subsection{Context}

The Aho, Hopcroft and Ullman (AHU) algorithm, introduced in the 1970s \cite[Example~3.2]{Aho1974}, establishes that the tree isomorphism problem can be solved in linear time, whereas the more general graph isomorphism problem is still an open problem today, where no proof of NP-completeness nor polynomial algorithm is known \cite{schoning1988graph}. However, the problem is considered to be solved in practice; powerful heuristics exist, such as the quasi-polynomial algorithm from \cite{babai2016graph}; see also \cite{mckay2014practical}.

As far as we know, AHU remains the only state-of-the-art algorithm for determining, in practice, whether two trees are isomorphic. Recently, Liu \cite{liu2021tree} proposed to represent a tree by a polynomial of two variables, computable in linear time, and where two trees have the same polynomial if and only if they are isomorphic. Unfortunately, the existence of an algorithm to determine the equality of two polynomials in polynomial time is still an open question \cite{saxena2009progress}. We should also mention \cite{buss1997alogtime}, which proposes an alternating logarithmic time algorithm for tree isomorphism -- under NC complexity class framework, that is, problems efficiently solvable on a parallel computer \cite{barrington1986bounded}.

However, one criticism -- emerging from Campbell and Radford in \cite{campbell1991tree} -- directed at the AHU algorithm is that it is presented in such a way that it is difficult to understand. We leave it to the reader to form their own opinion by reproducing the original text of the algorithm in Section~\ref{ss:ahu}, after a brief introduction of key background in Section~\ref{ss:tree_isom}. To the best of our knowledge, the remark from Campbell and Radford seems to have remained a dead letter in the community, and no alternative, clearer version of the algorithm seems ever to have been published -- with the exception of Campbell and Radford themselves,  which have nevertheless remained fairly close to the original text.

In this article, we propose to revisit the AHU algorithm by giving several alternative versions, all of them easier to understand and straightforward to implement. However, these variants have supra-linear complexity (which is also the case for the Campbell and Radford version). In practice, on trees of reasonable size ($\leq 10^6$), with a Python implementation using the \texttt{treex} library \cite{azais2019treex}, we find that two of the three proposed variants are faster than the original algorithm -- one of them sorts lists of integers (like the original algorithm), while the other replaces this step by calculating the product of a list of primes. We also propose a direct adaptation of our variants to compute tree compression into directed acyclic graphs (DAGs) \cite{godin2009quantifying} -- this time achieving state of the art complexity.

Section~\ref{ss:tree_isom} introduces the notations and definitions useful for the rest of the paper; the original AHU algorithm is presented in Section~\ref{ss:ahu}, as well as the aim of the paper.

\subsection{Tree isomorphisms}\label{ss:tree_isom}

A rooted tree $T$ is a connected directed graph without any undirected cycle such that (i) there exists a special node called the root and (ii) any node but the root has exactly one parent. The parent of a node $u$ is denoted by $\parent(u)$, whereas its children are denoted by $\children(u)$. The leaves $\leaves(T)$ of $T$ are the nodes without any children. Rooted trees are said to be unordered if the order among siblings is not significant; otherwise they are said to be ordered. This paper focuses only on unordered rooted trees, referred to simply as \emph{trees} in the remainder of this article.

The degree of a node $u$ is defined as $\deg(u)=\#\children(u)$ and the degree of a tree $T$ as $\deg(T) = \max_{u\in T}\deg(u)$. The depth $\depth(u)$ of a node $u$ is the length of the path between $u$ and the root. The depth $\depth(T)$ of $T$ is the maximal depth among all nodes. The level of a node $u$ is defined as $\depth(T)-\depth(u)$. The sets of nodes of level $d$ in a tree $T$ is denoted by $T^d$, and the mapping $d\mapsto T^d$ can be constructed in linear time by a simple traversal of $T$.

\begin{definition}
Two trees $T_1$ and $T_2$ are said to be \emph{isomorphic} if there exists a bijective mapping $\varphi : T_1\to T_2$ so that (i) the roots are mapped together and (ii) for any $u,v\in T_1$, $v\in\children(u) \iff \varphi(v)\in\children(\varphi(u))$.
\end{definition}

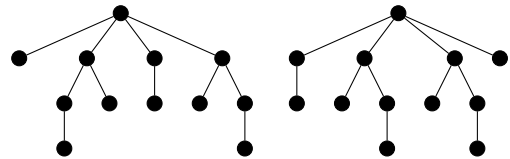
\begin{wrapfigure}[7]{R}{0.4\textwidth}
\vspace{-\baselineskip}
    \centering
\def\xscale{0.6}
\def\yscale{0.6}
\def\nodescale{0.6}
\begin{tikzpicture}[xscale=\xscale,yscale=\yscale]
\tikzstyle{fleche}=[-,>=latex]
\tikzstyle{noeud}=[draw,circle,fill=black,scale=\nodescale]

\node[noeud] (1) at (2.25,1) {};
\node[noeud] (2) at (0,0) {};
\node[noeud] (3) at (3,0) {};
\node[noeud] (4) at (1.5,0) {};
\node[noeud] (5) at (4.5,0) {};
\node[noeud] (6) at (3,-1) {};
\node[noeud] (7) at (1,-1) {};
\node[noeud] (8) at (2,-1) {};
\node[noeud] (9) at (4,-1) {};
\node[noeud] (10) at (5,-1) {};
\node[noeud] (11) at (1,-2) {};
\node[noeud] (12) at (5,-2) {};

\draw[fleche] (1)--(2);
\draw[fleche] (1)--(3);
\draw[fleche] (1)--(4);
\draw[fleche] (1)--(5);
\draw[fleche] (3)--(6);
\draw[fleche] (4)--(7);
\draw[fleche] (4)--(8);
\draw[fleche] (5)--(9);
\draw[fleche] (5)--(10);
\draw[fleche] (7)--(11);
\draw[fleche] (10)--(12);
\end{tikzpicture}\hfill 
\begin{tikzpicture}[xscale=\xscale,yscale=\yscale]
\tikzstyle{fleche}=[-,>=latex]
\tikzstyle{noeud}=[draw,circle,fill=black,scale=\nodescale]

\node[noeud] (1) at (2.25,1) {};
\node[noeud] (2) at (4.5,0) {};
\node[noeud] (3) at (0,0) {};
\node[noeud] (4) at (1.5,0) {};
\node[noeud] (5) at (3.5,0) {};
\node[noeud] (6) at (0,-1) {};
\node[noeud] (7) at (2,-1) {};
\node[noeud] (8) at (1,-1) {};
\node[noeud] (9) at (3,-1) {};
\node[noeud] (10) at (4,-1) {};
\node[noeud] (11) at (2,-2) {};
\node[noeud] (12) at (4,-2) {};

\draw[fleche] (1)--(2);
\draw[fleche] (1)--(3);
\draw[fleche] (1)--(4);
\draw[fleche] (1)--(5);
\draw[fleche] (3)--(6);
\draw[fleche] (4)--(7);
\draw[fleche] (4)--(8);
\draw[fleche] (5)--(9);
\draw[fleche] (5)--(10);
\draw[fleche] (7)--(11);
\draw[fleche] (10)--(12);
\end{tikzpicture}
    \captionof{figure}{Two isomorphic trees.}
    \label{fig:isomorphic_trees}
\end{wrapfigure}

Such a mapping $\varphi$ is called a \emph{tree isomorphism}. An example of isomorphic trees is provided in Figure~\ref{fig:isomorphic_trees}. Whenever two trees $T_1$ and $T_2$ are isomorphic, we note $T_1\simeq T_2$. It is well known that $\simeq$ is an equivalence relation on the set of trees \cite{valiente2002algorithms}. The \emph{tree isomorphism problem} consists in deciding whether two trees are isomorphic or not.

% \begin{minipage}[c]{0.5\textwidth}

% \end{minipage}\hfill
% \begin{minipage}[c]{0.45\textwidth}
% % \begin{figure}[H]

% % \end{figure}
% \end{minipage}

For the broader \emph{graph} isomorphism problem, it is not usual to explicitly construct the isomorphism  $\varphi$ -- let us mention nonetheless \cite[Section~3.3]{fortin1996graph} and \cite{ingels2021isomorphic,azais2023detection} -- but rather to compute a certificate of non-isomorphism. For instance, Weisfeiler-Lehman algorithms, also known as colour refinement algorithms \cite{huang2021short,kiefer2020power}, colour the nodes of each graph according to certain rules, and the histograms of the colour distributions are then compared: if they diverge, the graphs are not isomorphic. This test is not complete in the sense that there are non-isomorphic graphs with the same colour histogram -- even though the distinguishing power of these algorithms is constantly being improved \cite{grohe2021deep}. While the graph isomorphism problem is not solved in the general case, it is solved for trees by virtue of the AHU algorithm, which is built on a colouring principle similar to that of Weisfeiler-Lehman.

\subsection{The Aho, Hopcroft and Ullman algorithm}\label{ss:ahu}

We reproduce below the original text of the algorithm, as introduced in 1974 by Ahu, Hopcroft and Ullman in \cite[Example~3.2]{Aho1974} -- only minor changes have been made to fit the notations used in this paper.

\begin{enumerate}
    \item \cbstart\cbcolor{GreenBlack}First, assign to all leaves in $T_1$ and $T_2$ the integer $0$.
    \item Assume by induction that all nodes at level $d-1$ of $T_1$ and $T_2$ have been assigned an integer. Let $L_1$ (respectively $L_2$) be the list of nodes in $T_1$ (respectively $T_2$) at level $d-1$ sorted by non-decreasing value of the assigned integers.
    \item Assign to the nonleaves of $T_1$ at level $d$ a tuple of integers by scanning the list $L_1$ from left to right and performing the following actions:
    \begin{itemize}
        \item For each vertex on list $L_1$ take the integer assigned to $u$ to be the next component of the tuple associated with $\parent(u)$.
        \item On completion of this step, each nonleaf $w$ of $T_1$ at level $d$ will have a tuple $(i_1, i_2, \dots, i_k)$ associated with it, where $i_1, \dots, i_k$ are the integers, in non-decreasing order, associated with the children of $w$.
        \item Let $S_1$ be the sequence of tuples created for the vertices of $T_1$ on level $d$.
    \end{itemize}
    \item Repeat Step 3 for $T_2$ and let $S_2$ be the sequence of tuples created for the vertices of $T_2$ on level $d$.
    \item Sort $S_1$ and $S_2$ lexicographically. Let $S_1'$ and $S_2'$, respectively, be the sorted sequence of tuples.
    \item If $S_1'$ and $S_2'$ are not identical, then halt: the trees are not isomorphic. Otherwise, assign the integer $1$ to those vertices of $T_1$ on level $d$ represented by the first distinct tuple on $S_1'$, assign the integer $2$ to the vertices represented by the second distinct tuple, and so on. As these integers are assigned to the vertices of $T_1$ on level $d$, replace $L_1$ by the list of the vertices so assigned. Append the leaves of $T_1$ on level $d$ to the front of $L_1$. Do the same for $L_2$. $L_1$ and $L_2$ can now be used for the assignment of tuples to nodes at level $d+1$ by returning to Step 3.
    \item If the roots of $T_1$ and $T_2$ are assigned the same integer, $T_1$ and $T_2$ are isomorphic.\cbend
\end{enumerate}

Note that, in Step~$5$, the authors resort to a variant of radix sort \cite[Algorithm~3.2]{Aho1974}. Actually, the  tree isomorphism problem and AHU algorithm are only introduced in the book as an application example of this sorting algorithm. This algorithm can sort $n$ lists of varying lengths $l_1,\dots,l_n$, containing integers between $0$ and $m-1$, in complexity $O\left(\sum_{i=1}^n l_i + m\right)$. Expressed within our framework, the length of each list is exactly the degree of the associated node, and $m = \#\lbrace c(u):u\in T^d \rbrace$ -- where $c(u)$ designates the integer associated to node $u$.

\begin{theorem}[Aho, Hopcroft \& Ullman]
AHU algorithm runs in $O(n)$ where $n=\#T_1=\#T_2$.
\end{theorem}
\begin{proof}
See the proofs in \cite[Example~3.2]{Aho1974} for the whole algorithm and especially \cite[Algorithm~3.2]{Aho1974} for sorting lists $S_1$ and $S_2$ in Step~$5$. As stated above, Step~$5$ has complexity $O\left(\sum_{u\in T^d} \deg(u) + \#\lbrace c(u) : u\in T^d\rbrace\right)$. Noticing that $\sum_{u\in T^d} \deg(u) = \#T^{d-1}$ and that $\#\lbrace c(u) : u\in T^d\rbrace \leq \#T^d$, and summing over all levels indeed yields a linear complexity for all those sorts.
\end{proof}

\begin{remark}\label{rmk:size_of_trees}
 A point (very) briefly addressed by the authors of AHU algorithm specifies that the maximum integer $m$ used in Step~$5$ must be not ``too large'' \cite[Section 3.2, p.77]{Aho1974}. Indeed, the sorting algorithm works if the integers can actually be considered as integers, and not as sequences of $0$'s and $1$'s, as pointed out by Radford and Campbell \cite{campbell1991tree} -- in which they show that there are large trees for which the algorithms runs in $O(n\log n)$.

\medskip
\noindent
 How large are we talking? For the integers to \textbf{not} fit on one word of memory, we must assume that $m > 2^k$ with a $k$-bit machine. The smallest tree $T$ with $m=2^k$ has roughly $k2^k$ nodes (see Appendix~\ref{app:proof}). With $k=64$, this would imply trees of size $\approx 10^{21}$. For most practical applications, this is unlikely to be a problem.
\end{remark}

Without any additional context for interpreting what the algorithm does, perhaps the reader will agree with this comment, arguing that the formulation of the algorithm is

\begin{quote}
\textit{utterly opaque. Even on second or third reading. When an algorithm is written it should be
clear, it should persuade, and it should lend itself to analysis.}
\flushright{--- Douglas M. Campbell and David Radford \cite{campbell1991tree}}
\end{quote}

 In Campbell and Radford view, the formulation of the algorithm is detrimental to understanding it, analyse it, and implement it. It is true that the theoretical contribution of the AHU algorithm is indisputable, since it establishes that the tree isomorphism problem is linear. On the other hand, and to support Campbell and Radford's point of view, AHU would benefit from a formulation that is simpler to understand and implement. From a pedagogical point of view, it is likely that AHU is one of the first algorithms that people might want to study or implement when they learn about tree graph theory. AHU algorithm is also used by more advanced algorithms, such as the compression of trees into directed acyclic graphs (DAG) -- see Section~\ref{ss:global_classes}; or routinely for the construction of marked tree isomorphisms in \cite{ingels2021isomorphic,azais2023detection}.

\paragraph{Aim of the paper} In \cite{campbell1991tree}, Campbell and Radford provide a very clear, step-by-step exposition of the intuitions that lead to the AHU algorithm, and they even provide an algorithm similar to AHU that associates bitstrings to nodes instead of integers -- with $O(n\log n)$ time complexity. In this paper, we introduce yet another formulation for the AHU algorithm. This formulation assigns integers to the nodes, as does AHU and unlike Campbell and Radford's version. Several possible implementations of our approach are studied, both from a theoretical and a practical point of view. In addition to clarifying the intuition behind the original algorithm, our variants are straightforward to implement -- at the cost, however, of worse complexity. With respect to Remark~\ref{rmk:size_of_trees}, when used with trees of reasonable size in line with common use cases, they nonetheless perform better in practice. The outline of the paper is as follows:

\begin{itemize}
    \item Section~\ref{sec:revisiting} introduces our intuition for the AHU algorithm, in three variants: two using list sorts, and one using multiplication of lists of primes instead. We also present an adaptation of AHU for the compression of trees into directed acyclic graphs (DAG), also in three variants.
    \item Section~\ref{sec:num} tests these algorithms on simulated data of reasonable size, in competition with the original algorithm whenever possible (i.e. excluding DAG compression). 
\end{itemize}

Although the theoretical complexities of the algorithms presented here exceed the linear complexity of the original algorithm, we show that in practice, with the exception of one of the sorting variants, the others are competitive with the original. In particular, the variant using prime number multiplication is competitive with the best variant using sorts, even though it also has to generate the primes on the fly in addition to multiplying them.

Finally, in Appendix~\ref{app:proof}, we study a very specific class of trees, which allows us to prove results -- namely Lemmas~\ref{width_bound_pigeonhole} and~\ref{width_bound} -- relating the size of trees to the number of distinct integers needed to assign classes level by level in the AHU algorithm (whatever its variant). To the best of our knowledge, this matter has never been addressed before.

\section{Revisiting AHU algorithm}\label{sec:revisiting}

In this section, we present variants of the original AHU algorithm. First, Section~\ref{ss:intuition} provides a new intuition of the algorithm, in the form of a colouring process. We propose two variant implementations, each using a different sorting algorithm. In Section~\ref{ss:primes}, we replace the sorting step by the multiplication of a list of primes, leading to a new variant of AHU. Finally, in Section~\ref{ss:global_classes}, we focus on the compression of trees into directed acyclic graphs (DAGs), which can be achieved via a simple modification of our version of AHU -- declined in three implementations: two with sorting, one with prime multiplication, as before.

\subsection{An intuition for AHU and two variants}\label{ss:intuition}

As already stated, the interested reader can found in \cite{campbell1991tree} a step-by-step explaination of the concepts at works behind AHU algorithm. Here, we introduce another intuition for the AHU algorithm, presented as a colouring process, thus making the connection with Weisfeiler-Lehman algorithms for graph isomorphism already mentioned.

The core idea behind AHU algorithm is to provide each node in trees $T_1$ and $T_2$ a canonical representative of its equivalence class for $\simeq$, thus containing all the information about its descendants.

% \begin{minipage}[c]{0.4\textwidth}
The nodes of both trees are simultaneously browsed in ascending levels. Suppose that each node $u$ of level $d-1$ has been assigned a colour $c(u)$, representing its equivalence class for the relation $\simeq$. Each node $u$ of level $d$ is associated with a multiset $\multiset(u)=\lbrace\!\lbrace c(v) : v\in \children(u)\rbrace\!\rbrace$ -- if $u$ is a leaf, this multiset is denoted $\emptyset$. Each distinct multiset is given a colour, which is assigned to the corresponding nodes. An illustration is provided in Figure~\ref{fig:AHU_principle}. In the end, the trees are isomorphic if and only if their roots receive the same colour. Moreover, after processing level $d$, if the multiset of colours assigned to the nodes of level $d$ differs from one tree to the other, we can immediately conclude that the trees are not isomorphic.
% \end{minipage}\hfill
% \begin{minipage}[c]{0.55\textwidth}

% \begin{figure}[H]
\begin{wrapfigure}[17]{R}{0.55\textwidth}
\vspace{-\baselineskip}
\centering
\def\xscale{0.7}
\def\yscale{1}
\def\nodescale{1}
\begin{tikzpicture}[xscale=\xscale,yscale=\yscale]
\tikzstyle{fleche}=[-,>=latex]
\tikzstyle{noeud}=[draw,circle,fill=black,scale=\nodescale]

\draw[dashed,lightgray,ultra thick] (-3,-2.75)--(9,-2.75);
\draw[dashed,lightgray,ultra thick] (-3,-1.75)--(9,-1.75);
\draw[dashed,lightgray,ultra thick] (-3,-0.25)--(9,-0.25);
\draw (3,0.5)--(9,0.5);

\node at (3.5,0.75) {Level};
\node at (3.5,0) {$3$};
\node at (3.5,-1) {$2$};
\node at (3.5,-2) {$1$};
\node at (3.5,-3) {$0$};

\node at (6,0.75) {Multiset};
\node at (6,-3) {$\emptyset$};
\node at (6,-2.5) {$\emptyset$};
\node at (6,-2) {$\lbrace\!\lbrace \node[lgreen] \rbrace\!\rbrace$};
\node at (6,-1.5) {$\emptyset$};
\node at (6,-1) {$\lbrace\!\lbrace \node[lgreen] \rbrace\!\rbrace$};
\node at (6,-0.5) {$\lbrace \!\lbrace\node[lblue],\node[lgreen]\rbrace\!\rbrace$};
\node at (6,0) {$\lbrace\!\lbrace \node[lred],\node[lred],\node[lblue],\node[lgreen]\rbrace\!\rbrace$};

\node at (8,0.75) {Colour};
\node at (8,-3) {$\node[lgreen]$};
\node at (8,-2.5) {$\node[lgreen]$};
\node at (8,-2) {$\node[lblue]$};
\node at (8,-1.5) {$\node[lgreen]$};
\node at (8,-1) {$\node[lblue]$};
\node at (8,-0.5) {$\node[lred]$};
\node at (8,0) {$\phantom{^\star}\node[lorange]^\star$};

\def\x{-2.5};
\def\y{-1};

\node[noeud,fill=lorange] (1) at (2.25+\x,1+\y) {};
\node[noeud,fill=lgreen] (2) at (0+\x,0+\y) {};
\node[noeud,fill=lblue] (3) at (3+\x,0+\y) {};
\node[noeud,fill=lred] (4) at (1.5+\x,0+\y) {};
\node[noeud,fill=lred] (5) at (4.5+\x,0+\y) {};
\node[noeud,fill=lgreen] (6) at (3+\x,-1+\y) {};
\node[noeud,fill=lblue] (7) at (1+\x,-1+\y) {};
\node[noeud,fill=lgreen] (8) at (2+\x,-1+\y) {};
\node[noeud,fill=lgreen] (9) at (4+\x,-1+\y) {};
\node[noeud,fill=lblue] (10) at (5+\x,-1+\y) {};
\node[noeud,fill=lgreen] (11) at (1+\x,-2+\y) {};
\node[noeud,fill=lgreen] (12) at (5+\x,-2+\y) {};

\draw[fleche] (1)--(2);
\draw[fleche] (1)--(3);
\draw[fleche] (1)--(4);
\draw[fleche] (1)--(5);
\draw[fleche] (3)--(6);
\draw[fleche] (4)--(7);
\draw[fleche] (4)--(8);
\draw[fleche] (5)--(9);
\draw[fleche] (5)--(10);
\draw[fleche] (7)--(11);
\draw[fleche] (10)--(12);

\end{tikzpicture}
\captionof{figure}{Assigning colours to nodes in AHU algorithm. $\star$:~We could have used colour \node[lgreen] because colours are assigned level by level and not globally. We have chosen to use another colour for the sake of clarity. In this example, the colours correspond exactly to the equivalence classes of the nodes.}
    \label{fig:AHU_principle}

\end{wrapfigure}
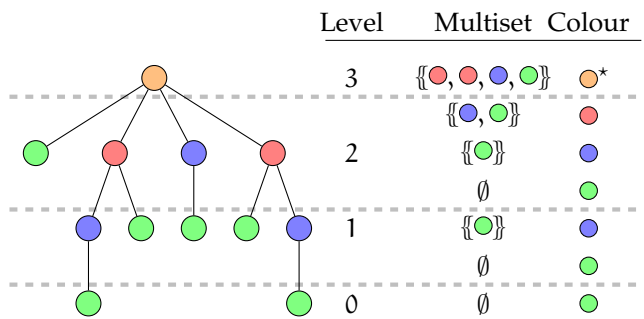

% \end{figure}
% \end{minipage}

Consider the number of colours required by any version of AHU algorithm; this number is given by
\begin{equation}\label{eq:width}
\max_{d\in [\![0,\depth(T)]\!]} \#\lbrace c(u): u\in T^d\rbrace.
\end{equation} 
We call it the \emph{width} of $T$ and denote it by $\width(T)$.

In practice, colours are represented by integers. To associate different integers with distinct multisets, we need to keep track of which ones we have already encountered. In order to check in constant time whether a multiset has already been seen (which is the case in the original algorithm: as the tuples are sorted in Step~$5$, it is enough to compare a tuple with its predecessor in the list $S'_i$ to find out whether they are different or not), we need a perfect hash function that works on multisets. Obviously, this is a very strong assumption. Hash functions for multisets do exist -- see for instance \cite{clarke2003incremental,maitin2017elliptic} -- but they involve advanced concepts, which would make implementation difficult for non-specialists. For the sake of this article, let us assume that we do not have access to such methods. Since we will be focusing on Python applications later on, we assume that the Python dictionary structure can be seen as a perfect hash table; it can hash integers, strings or tuples. 

To get around multisets, a simple solution is to see them as lists, which we sort before hashing them as tuples. This approach, in particular, is used in Weisfeiler-Lehman algorithms, where the same problem arises -- see, for example, \cite[Algorithm~3.1]{kiefer2020power}. The pseudocode for AHU as presented in this section, using prior sorting of multisets, is presented in Algorithm~\ref{algo:ideal_ahu}.

\begin{minipage}[c]{0.4\textwidth}
Note that to compare the multisets of integers associated at current level $d$, on line~$14$, it is not necessary to hash but simply to sort the two lists and compare them term by term. This can be accomplished via pigeonhole sort \cite{black1998dictionary}. Remember that pigeonhole sorting a list of $k$ integers within the range $0$ to $m-1$ is done in $O(k+m)$. Here, since colours are attributed at each level, $m\leq \#T_i^d$; hence a complexity of $O(\#T_i^d)$ for this step.

\medskip

The overall complexity of Algorithm~\ref{algo:ideal_ahu} depends on the sorting algorithm used in line~$9$ to sort $\multiset(u)$. It may be tempting to reuse the pigeonhole sort, already used in the algorithm, or to use a comparison sorting algorithm, such as timsort, Python's native algorithm \cite{auger2018worst}. We assume that $T_1=T_2=T$ -- this is the worst case, since, if $T_1\not\simeq T_2$, we do not visit all the levels.

\end{minipage}\hfill
\begin{minipage}[c]{0.55\textwidth}

% \begin{wrapfigure}[17]{R}{0.51\textwidth}
\begin{algorithm}[H]
\caption{\textsc{TreeIsomorphism}}
\label{algo:ideal_ahu}
\KwIn{$T_1,T_2$}
\KwOut{$\top$ if and only if $T_1\simeq T_2$}
\SetKw{KwFr}{from}
\SetKw{KwAnd}{and}
\eIf{$\depth(T_1)\neq\depth(T_2)$}{\Return $\bot$}{
\For{$d$ \KwFr $0$ \KwTo $\depth(T_1)$}{
$k \gets 0$\;
Let $f : \emptyset \mapsto 0$\;
\For{$i\in\lbrace 1,2\rbrace$ \KwAnd $u\in T_i^d$}{
$\multiset(u)\gets (c(v) : v\in \children(u))$\;
Sort $\multiset(u)$\;
    \If{$f(\multiset(u))$ \emph{is not defined}}{
        $k\gets k+1$\;
        Define $f(\multiset(u)) = k$\;
    }
 $c(u) \gets f(\multiset(u))$\;
}
\If{$\lbrace\!\lbrace c(u) : u\in T_1^d\rbrace\!\rbrace \neq \lbrace\!\lbrace c(u) : u\in T_2^d\rbrace\!\rbrace $}{\Return $\bot$}
}
\Return $\top$}
\end{algorithm}
% \end{wrapfigure}
\end{minipage}

\begin{proposition}
Algorithm~\ref{algo:ideal_ahu} runs
\begin{enumerate}[label=(\roman*)]
    \item in $O\left( \displaystyle\frac{\#T^2}{\log\# T}\right)$ using pigeonhole sort;
    \item in $O(\#T\log\deg(T))$ using timsort.
\end{enumerate}
\end{proposition}
\begin{proof}
Fix a level $d$ and a node $u\in T_i^d$. Building $\multiset(u)$ requires $O(\deg(u))$. Then, sorting $\multiset(u)$ depends on the algorithm used.
\begin{enumerate}[label=(\roman*)]
    \item Pigeonhole sort is done in $O(\deg(u) + \max(\multiset(u)))$.  Since the colours are attributed from $0$ to $\width(T)-1$ -- recall $\width(T)$ from \eqref{eq:width}, we have $\max(\multiset(u))\leq \width(T)$ and therefore a complexity of $O(\deg(u)+\width(T))$.
    \item The worst-case complexity of timsort is $O(\deg(u) \log\deg(u))$ \cite{auger2018worst}.
\end{enumerate}

Notice that $\sum_{u\in T_i^d }\deg(u) = \#T_i^{d-1}$. Recalling that line~$14$ is processed in $O(\#T^d_i)$ via pigeonhole sort, the complexity of treating level $d$ is $O(\#T_i^d \width(T))$ for case (i), and $O(\#T_i^d \log\deg(T))$ for case (ii) -- using $\log\deg(u)\leq \log\deg(T)$. Summing over $d$ leads to the result for (ii), and to $O(\#T\width(T))$ for (i).

The results holds in case (i) by virtue of the following lemma, whose proof can be found in Appendix~\ref{app:proof}.\end{proof}

\begin{lemma}\label{width_bound_pigeonhole}
For any tree $T$, $\width(T) = O\left(\displaystyle\frac{\#T}{\log\#T}\right)$.
\end{lemma}

Neither version of Algorithm~\ref{algo:ideal_ahu} is linear; we will see later in Section~\ref{sec:num} how they behave in practice. In the next section, we will consider another approach, which does not resort to sorting, but instead replaces multiset hashing with prime number multiplication.

\begin{remark}\label{rmk:ordered_labelled}
It should be noted that Algorithm~\ref{algo:ideal_ahu} can be straightforwardly adapted to handle ordered and/or labelled trees (where each node carries a label). For ordered trees, it suffices to not sort $\multiset(u)$. For labelled trees, we have to assume that labels can be hashed. We replace $\multiset(u)$ by the tuple $(\text{label}(u),\multiset(u))$ and consider two tuples to be equal if both the label and the multiset are identical. Note that another way of handling labels exists, requiring not the equality of labels but rather the respect of label equivalence classes. This variant is known as marked tree isomorphism \cite{booth1979problems} and proved to be as hard as graph isomorphism -- for which no polynomial algorithm is known, even though in practice very efficient algorithms exist; see \cite{mckay2014practical,babai2016graph} for instance. Marked tree isomorphism is far beyond the scope of this article; but we refer the interested reader to \cite{azais2023detection,ingels2021isomorphic}.

\medskip
\noindent
The authors of the AHU algorithm have provided in \cite{Aho1974} a way of adapting their method to labeled trees - but only with labels that can be totally ordered. It suffices to add of label of node $u$ as the first element of the tuple associated to it, before the lexicographical sort of Step~$5$. While not provided in \cite{Aho1974}, their algorithm can also be modified to account for ordered trees. In Step~$3$, the tuple associated to node $u$ at level $d$ is instead calculated as the tuple of integers associated with its children, in order. The list $L_i$ is not longer necessary.
\end{remark}
\subsection{AHU with primes}\label{ss:primes}

In Algorithm~\ref{algo:ideal_ahu}, we need to associate a unique integer $f(\multiset(u))$ to each distinct multiset $\multiset(u)$ of integers encountered. There is a particularly simple and fundamental example where integers are associated with multisets: prime factorization. Indeed, via the fundamental theorem of arithmetic, there is a bijection between integers and multisets of primes. For example, $360 = 2^3\cdot 3^2 \cdot 5$ is associated to the multiset $\lbrace\!\lbrace 2,2,2,3,3,5\rbrace\!\rbrace$.

\begin{minipage}[c]{0.4\textwidth}
 Note that this bijection is well known \cite{blizard1989multiset}, and has already been successfully exploited in the literature for prime decomposition, but also usual operations such as product, division, gcd and lcd of numbers \cite{tarau2014towards}. To the best of our knowledge, this link has never been exploited to replace multiset hashing, a fortiori in the context of graph isomorphism algorithms -- such as Weisfeiler-Lehman, or AHU for trees. Note, however, that this approach has been used in the context of evaluating poker hands \cite{suffecool2005cactus}, where prime multiplication has been preferred to sorting cards by value in order to get a unique identifier for each distinct possible hand.

\medskip

Since the previous versions of AHU we presented (both the original and our variants) sort lists of integers, the main challenge of this substitution concerns the potential additional complexity of multiplying lists of primes compared to sorting lists of integers.

\medskip

\end{minipage}\hfill
\begin{minipage}[c]{0.55\textwidth}
\vspace{-\baselineskip}
\begin{algorithm}[H]
\caption{\textsc{TreeIsomorphismWithPrimes}}
\label{algo:AHU_primes}
\KwIn{$T_1,T_2$}
\KwOut{$\top$ if and only if $T_1\simeq T_2$}
\SetKw{KwFr}{from}
\SetKw{KwAnd}{and}
\eIf{$\depth(T_1)\neq\depth(T_2)$}{\Return $\bot$}{
$P \gets [2,3,5,7,11,13]$ and $N_{\text{sieve}}\gets 16$\;
\For{$d$ \KwFr $0$ \KwTo $\depth(T_1)$}{
Let $f: 1\mapsto 2$\;
$p\gets 2$\;
\For{$i\in\lbrace 1,2\rbrace$ \KwAnd $u\in T_i^d$}{
$N(u)\gets \displaystyle\prod_{v\in\children(u)} c(v)$\;
\If{$f(N(u))$ \emph{is not defined}}{
$N_{\text{sieve}},P,p\gets \textsc{NextPrime}(N_{\text{sieve}},P,p)$\;
Define $f(N(u))=p$\;
}
$c(u)\gets f(N(u))$\;
}
\If{$\lbrace\!\lbrace c(u) : u\in T_1^d\rbrace\!\rbrace \neq \lbrace\!\lbrace c(u) : u\in T_2^d\rbrace\!\rbrace $}{\Return $\bot$}
}
\Return $\top$}
\end{algorithm}
\end{minipage}

Suppose that each node $u$ at level $d$ has received a prime number $c(u)$, assuming that all nodes at that level and of the same class of equivalence have received the same number. Then, to a node $u$ at level $d$, instead of associating the multiset $\multiset(u)=\lbrace\!\lbrace c(v) : v\in  \children(u)\rbrace\!\rbrace$, we associate the number $N(u) = \prod_{v\in \children(u)} c(v)$. The nodes of level $d$ are then renumbered with prime numbers  -- where each distinct number $N(u)$ gets a distinct prime. The fundamental theorem of arithmetic ensures that two identical multisets $\multiset(\cdot)$ receive the same number $N(\cdot)$. The pseudocode for this new version of AHU is presented in Algorithm~\ref{algo:AHU_primes}.

% \begin{wrapfigure}[23]{R}{0.55\textwidth}

% \end{wrapfigure}

% \begin{minipage}[c]{0.4\textwidth}

% \medskip

The subroutine \textsc{NextPrime}, introduced in Algorithm~\ref{algo:next_prime}, returns the next prime not already used at the current level; if there is no unassigned prime in the current prime list $P$, then new primes are generated using a segmented version of the sieve of Eratosthenes.

% \end{minipage}\hfill
% \begin{minipage}[c]{0.55\textwidth}

% \end{minipage}

\begin{algorithm}[h!]
\caption{\textsc{NextPrime}}
\label{algo:next_prime}
\KwIn{$N_{\text{sieve}}$, the largest number for which the sieve of Eratosthenes has already been performed, the list $P$ of primes $\leq N_{\text{sieve}}$ in ascending order, with $\length(P)\geq 6$, and a prime $p\in P$}
\If{$p$ \emph{is the last of element of} $P$}{
$N_{\text{sieve}},P \gets \textsc{SieveOfEratosthenes}(N_{\text{sieve}},P)$\; 
\tcc{At least one new prime has been added to $P$}
}
Let $p'$ be the next prime after $p$ in $P$\;
\Return $N_{\text{sieve}},P,p'$
\end{algorithm}

Let us denote $p_n$ the $n$-th prime number. There are well known bounds on the value of $p_n$ \cite{dusart1999k,rosser1941explicit} -- with $\ln$ denoting the natural logarithm and $n\geq 6$:
\begin{equation}\label{primes_bound}
n (\ln n + \ln\ln n -1)  < p_n < n (\ln n + \ln \ln n).
\end{equation}

Suppose we have the list of all primes $P\leq N_{\text{sieve}}$, where $N_{\text{sieve}}$ is the largest integer sieved so far. With $\#P=n-1$, to generate $p_n$, we simply resume the sieve up to the integer $\lceil n(\ln n+\ln\ln n)\rceil$, starting from $\lfloor n(\ln n + \ln\ln n-1)\rfloor$ or $N_{\text{sieve}}+1$, whichever is greater -- to make sure there is no overlap between two consecutive segments of the sieve. With this precaution in mind, the total complexity of the segmented sieve is the same as if we had directly performed the sieve in one go \cite{bays1977segmented}; i.e., $O(N\log\log N)$ for a sieve performed up to integer $N$. Therefore, to generate the first $n$ prime numbers, according to \eqref{primes_bound}, we have $N= n(\ln n + \ln\ln n)=O(n\log n)$ and the final complexity of the sieve can be evaluated as $O(n \log n \log\log n)$. See \cite{o2009genuine} for practical considerations on the implementation of the segmented sieve of Eratosthenes.

\begin{remark}
Note that other sieve algorithms exist, with better complexities -- such as Atkin sieve \cite{atkin2004prime} or the wheel sieve \cite{pritchard1987linear}; the sieve of Eratosthenes has the merit of being the simplest to implement and sufficient for our needs. Also, a better asymptotic complexity but with a worse constant can be counterproductive for producing small primes -- which is rather our case since we generate the primes in order.
\end{remark}

We now analyse the complexity of Algorithm~\ref{algo:AHU_primes}, assuming  that $T_1=T_2=T$. Following the previous discussion, we can consider separately the complexity for generating the primes numbers.

\begin{proposition}\label{prop:primes_generation}
Generating the primes required for Algorithm~\ref{algo:AHU_primes} can be done in $O(\#T\log\log T)$.
\end{proposition}
\begin{proof}
To generate the first $n$ primes, the sieve must be carried out up to the integer $n\cdot(\ln n +\ln\ln n)$, for total complexity $O(n\log n\log\log n)$. As defined in \eqref{eq:width}, the number of primes needed by Algorithm~\ref{algo:AHU_primes} is equal to $\width(T)$. We result immediately follows by virtue of the upcoming lemma, whose proof can be found in Appendix~\ref{app:proof}.
\end{proof}
\begin{lemma}\label{width_bound}
For any tree $T$, $\width(T)\ln \width(T) = O(\#T)$.
\end{lemma}
Finally, we have the following result.

\begin{theorem}\label{th:complexity_primes}
Algorithm~\ref{algo:AHU_primes} runs in 
$O\left(\#T \log \#T \cdot M\big(\deg(T) \log\#T\big)\right)$ where $M(n)$ varies from $\log n$ to $n$ depending on the multiplication algorithm used.
\end{theorem}
\begin{proof}
The proof can be found in Appendix~\ref{app:proof_complexity_primes}.
\end{proof}

This new variant is no more linear than the two versions introduced in Section~\ref{ss:intuition}. We shall see in Section~\ref{sec:num} that, in practice, this version is competitive with its rivals.

\subsection{Computing classes globally for DAG compression}\label{ss:global_classes}
In AHU algorithm as it has been presented so far, the colours assigned to the nodes are assigned level by level, and therefore make it possible to determine the equivalence class of the node relative to the level at which it is located. It is legitimate to ask whether the colours can be associated globally, so that the colour of a node is an exact reflection of its equivalence class in the tree -- see Figure~\ref{fig:AHU_principle}. In this way, two nodes located at different levels but having the same colour induce isomorphic subtrees.

The need to assign equivalence classes globally notably arises when considering the (lossless) compression of trees into directed acyclic graphs (DAG). Trees can have many redundancies in their structure, and the aim of DAG compression is precisely to eliminate these redundancies. There are many applications of DAG compression, some of which are: the representation of trees in computer graphics \cite{sutherland1963sketchpad,hart1991efficient}, the simplification of queries on XML documents \cite{buneman2003path,frick2003query} or again the computation of convolution kernels \cite{aiolli2006fast,azais2020weight}.

\begin{minipage}[c]{0.45\textwidth}
The set of vertices of the DAG compression $\red(T)$ of a tree $T$  corresponds to the set of equivalence classes of the subtrees of $T$. For any vertices $a,b\in \red(T)$, the multiplicity of the arc $a\to~b$ corresponds to the number of children of class $b$ of a subtree of class $a$.  An example can be found in Figure~\ref{fig:dag_compression_example}. A more precise definition and an algorithm can be found in \cite{godin2009quantifying}. However, it should be noted that a simple adaptation of Algorithm~\ref{algo:ideal_ahu} can also be used to construct the DAG compression of a tree, namely Algorithm~\ref{algo:dag_compression}.
\end{minipage}\hfill
\begin{minipage}[c]{0.5\textwidth}
% \begin{figure}[H]
%     \centering
% \begin{subfigure}[t]{0.45\textwidth}
\centering
\def\xscale{0.7}
\def\yscale{1}
\def\nodescale{1}
\begin{tikzpicture}[xscale=\xscale,yscale=\yscale]
\tikzstyle{fleche}=[-,>=latex]
\tikzstyle{noeud}=[draw,circle,fill=black,scale=\nodescale]

\def\x{-2.5};
\def\y{-1};

\node[noeud,fill=lorange] (1) at (2.25+\x,1+\y) {};
\node[noeud,fill=lgreen] (2) at (0+\x,0+\y) {};
\node[noeud,fill=lblue] (3) at (3+\x,0+\y) {};
\node[noeud,fill=lred] (4) at (1.5+\x,0+\y) {};
\node[noeud,fill=lred] (5) at (4.5+\x,0+\y) {};
\node[noeud,fill=lgreen] (6) at (3+\x,-1+\y) {};
\node[noeud,fill=lblue] (7) at (1+\x,-1+\y) {};
\node[noeud,fill=lgreen] (8) at (2+\x,-1+\y) {};
\node[noeud,fill=lgreen] (9) at (4+\x,-1+\y) {};
\node[noeud,fill=lblue] (10) at (5+\x,-1+\y) {};
\node[noeud,fill=lgreen] (11) at (1+\x,-2+\y) {};
\node[noeud,fill=lgreen] (12) at (5+\x,-2+\y) {};

\draw[fleche] (1)--(2);
\draw[fleche] (1)--(3);
\draw[fleche] (1)--(4);
\draw[fleche] (1)--(5);
\draw[fleche] (3)--(6);
\draw[fleche] (4)--(7);
\draw[fleche] (4)--(8);
\draw[fleche] (5)--(9);
\draw[fleche] (5)--(10);
\draw[fleche] (7)--(11);
\draw[fleche] (10)--(12);

\def\x{6}
\node[noeud,fill=lorange] (a) at ({\x},{0}) {};
\node[noeud,fill=lred] (b) at ({\x},{-1.0}) {};
\node[noeud,fill=lblue] (c) at ({\x},{-2}) {};
\node[noeud,fill=lgreen] (d) at ({\x},{-3.0}) {};

\tikzstyle{arc}=[->-,>=latex]

\draw[arc] (a) to (b) ;
\draw[arc] (a) to[bend left=60] (b) ;
\draw[arc] (a) to[bend right=45] (c);
\draw[arc] (c) to (d);
\draw[arc] (b) to (c);
\draw[arc] (b) to[bend left=45] (d);
\draw[arc] (a) to [bend right=90](d);

\end{tikzpicture}
\captionof{figure}{A tree $T$ (left) and its DAG compression $\red(T)$ (right). Nodes are colored according to their equivalence class under $\simeq$.}
\label{fig:dag_compression_example}
% \end{figure}
\end{minipage}

\begin{minipage}[c]{0.39\textwidth}

\begin{remark}
In light of Remark~\ref{rmk:ordered_labelled}, the same adjustments can be applied to Algorithm~\ref{algo:dag_compression} to take into account ordered and/or labelled trees. In addition, for labelled trees, the vertex in $Q$ associated to the tuple $(\text{label}(u),\multiset(u))$ is labelled with $\text{label}(u)$; $f$ must be initialized as an empty mapping, and $Q$ as an empty graph. For ordered trees, the order of children of $\multiset(u)$ must be respected in $Q$.
\end{remark}

As with Algorithm~\ref{algo:ideal_ahu}, the complexity of Algorithm~\ref{algo:dag_compression} depends on the sorting algorithm used.

\begin{proposition}
Algorithm~\ref{algo:dag_compression} computes $\red(T)$
\begin{enumerate}[label=(\roman*)]
    \item in $O\left(\#T^2\right)$ using pigeonhole sort;
    \item in $O(\#T\log\deg(T))$ using timsort.
\end{enumerate}
\end{proposition}

\end{minipage}\hfill
\begin{minipage}[c]{0.56\textwidth}
\begin{algorithm}[H]
\caption{\textsc{DAGCompression}}
\label{algo:dag_compression}
\KwIn{$T$}
\KwOut{$\red(T)$}
\SetKw{KwFr}{from}
\SetKw{KwAnd}{and}
$k \gets 0$\;
Let $f : \emptyset \mapsto 0$\;
Let $Q$ be a graph with a unique vertex $0$\;
\For{$d$ \KwFr $0$ \KwTo $\depth(T)$}{
\For{$u\in T^d$}{
    $\multiset(u)\gets (c(v) : v\in\children(u))$\;
Sort $\multiset(u)$\;
    \If{$f(\multiset(u))$ \emph{is not defined}}{
        $k\gets k+1$\;
        Define $f(\multiset(u)) = k$\;
        Create a vertex $k$ in $Q$ with children $\multiset(u)$\;
    }
 $c(u) \gets f(\multiset(u))$\;
}
}
\Return $Q$
\end{algorithm} 
\end{minipage}

\begin{proof}
We assume that creating a new vertex can be done in $O(\#\multiset(u))=O(\deg(u)$. Then, note that since we can determine whether $f(\multiset(u)$ is defined in $O(1)$ at line~$8$, the complexity of constructing the DAG is bounded by $ O\left(\sum_{u\in T} \deg(u)\right) = O(\#T)$ (the worst case would be when all nodes of $T$ are of different equivalence class), and can be counted independently from other steps in the algorithm.

With timsort, the sorting complexity remains $O(\deg(u)\log\deg(u))$ and the global complexity is the same as for Algorithm~\ref{algo:ideal_ahu}. However, for pigeonhole sort, things are diffrent. By computing the classes globally, the number of colours needed is not $\width(T)$ but exactly $\#\red(T)$, that can be roughly bounded by $\#T$; therefore sorting $\multiset(u)$ via pigeonhole has complexity $O(\deg(u)+\#T)$ and leads to a global complexity of $O(\#T^2)$.
\end{proof}

The state-of-the-art complexity for computing $\red(T)$ is $O(\#T\log\deg(T))$\footnote{Actually, in \cite{godin2009quantifying}, the authors announce a complexity of $O(\#T\deg(T)\log\deg(T))$; but their proof does not exploit the fact that $\sum_{u\in T}\deg(u)=\#T-1$. Taking this into account, we can simplify their result and obtain exactly $O(\#T\log\deg(T))$.} \cite[Section~2.2.1]{godin2009quantifying}, so our approach is consistent if we use timsort.

Algorithm~\ref{algo:AHU_primes} can also be adapted so that the DAG is constructed using multiplication of primes. As stated above, the largest colour (thus prime) used in the algorithm is no longer $\width(T)$, but $\#\red(T)$. Using the same rough bound of $\#\red(T)\leq \#T$, the complexity of Algorithm~\ref{algo:AHU_primes} would be modified to $$O\big(\#T \log \#T \cdot \left(M\big(\deg(T) \log\#T\big)+\log\log\#T\right)\big),$$
the final complexity depending on the chosen multiplication algorithm and whether $M(\deg(T)\log\#T)$ outweighs or not $\log\log\#T$ -- recalling that $M(n)$ varies from $\log n$ to $n$.

\begin{remark}
Concerning the bound $\#\red(T)\leq \#T$, note that $\#\red(T)=\#T$ is achieved if and only if $T$ is a chain; in which case all nodes have at most one child, and only one prime is actually needed. We could refine the bound with $\#\red(T)\leq \#T - \#\leaves(T) +1$, since all leaves of $T$ have the same equivalence class. In general, if we choose a tree $T$ uniformly at random among trees with $n$ nodes, the expected value of $\red(T)$ is $$\sqrt{\frac{\ln 4}{\pi}} \frac{n}{\sqrt{\ln n}}\left(1+O\left(\frac{1}{\ln n}\right)\right) \quad \text{\cite[Theorem~29]{bousquet2015xml}}.$$

\noindent
The most compressible trees are called \emph{self-nested trees} and achieve $\#\red(T)=\depth(T)$ -- on that matter, see \cite{godin2009quantifying, azais2019nearest, azais2019approximation}.
\end{remark}

Finally, note that AHU as stated in the original paper \cite{Aho1974} cannot be extended to take into account this global assignment of colours.

\section{Numerical experiments}\label{sec:num}

% \textcolor{red}{chapeau}

% \subsection{Comments on the implementation}

Table~\ref{tab:recap} summarizes the different algorithms seen so far and whether or not they can be adapted to compute DAG compressions of trees. We implemented in Python those 7 different algorithms. All of them are written to be fully compatible with the library \verb=treex= \cite{azais2019treex}, dedicated to tree algorithms.

\begin{table}[h!]
{
\small
\centering
\begin{tabular}{c|c|c|c|c}
    & Original & AHU with & AHU with &  \multirow{2}{*}{AHU with primes}  \\
     & AHU &   pigeonhole sort&   timsort& \\
     \hline
            &&&&   \\[-1em]
Tree & \multirow{2}{*}{ $O(\#T)$ }  & \multirow{2}{*}{ $O\left(\displaystyle\frac{\#T^2}{\log\#T}\right)$} & \multirow{2}{*}{$O(\#T\log\deg(T))$  } & \multirow{2}{*}{$O(\#T\log\#T \cdot M\left(\deg(T) \log\#T\right)))$}\\
Isomorphism &&&&\\
       &&&&   \\[-1em]
DAG &\multirow{2}{*}{\xmark}  & \multirow{2}{*}{$O(\#T^2)$} & \multirow{2}{*}{$O(\#T\log\deg(T))$ } & \multirow{2}{*}{$O\big(\#T \log \#T \cdot \left(M\big(\deg(T) \log\#T\big)+\log\log\#T\right)\big)$}\\
Compression &&&&\\
\end{tabular}
}
    \caption{Complexities of the various algorithms studied in this paper.}
    \label{tab:recap}
\end{table}

Most of auxiliary functions used in those algorithms have been implemented as well, such as the variant of radix sort used by AHU \cite[Algorihm~3.2]{Aho1974}, pigeonhole sort or the prime variant of pigeonhole sort discussed in Appendix~\ref{app:proof_complexity_primes}. For the algorithm using primes, we have also implemented a variant that uses a pregenerated list of primes, large enough so that no additional sieving step is needed during execution -- in this way we intend to study the impact of multiplication compared with sorting. The divide and conquer recursive multiplication strategy discussed in Appendix~\ref{app:proof_complexity_primes} as well as the sieve of Eratosthenes have been also implemented. Note that our implementation of the segmented sieve of Eratosthenes ignores multiples of $2$ and $3$, thus making the sieve $6$ times faster. The noteworthy native Python functionalities we used are (i) the dictionary structures for hash tables, (ii) the \verb=sort= method for lists -- that use timsort algorithm \cite{auger2018worst} and (iii) the multiplication operator \verb=*= -- which use schoolbook multiplication for small integers, and Karatsuba for large ones. 

All experiments have been conducted on a HP Elite Notebook with 32 Go of RAM and Intel Core i7-1365U processor.

% \textcolor{red}{les scripts sont accessibles sur github (?)}

% \textcolor{red}{plutôt que "juste" des boxplot, comme on teste chaque algo sur chaque arbre, on peut sortir l'histogramme du nb de fois où X a battu Y, etc. En éliminant pigeonhole il y a 6 permutations à regarder}

% \textcolor{red}{représentations:
% "Saari, Donald G. (1999). Explaining All Three-Alternative Voting Outcomes. Journal of Economic Theory, 87, 313-355."
% \url{https://maa.org/book/export/html/115945}}

% \fbox{\begin{tikzpicture}

% \def\x{0.75}

% \node[label=0:$A$] (A) at (-30:2*\x) {$\bullet$};
% \node[label=90:$B$] (B) at (120-30:2*\x) {$\bullet$};
% \node[label=180:$C$] (C) at (240-30:2*\x) {$\bullet$};

% \node[label=0:$D$] (D) at (30:\x) {$\bullet$};
% \node[label=180:$E$] (E) at (120+30:\x) {$\bullet$};
% \node[label=-90:$F$] (F) at (240+30:\x) {$\bullet$};

% \node (O) at (0:0) {$\bullet$};

% \node (2) at (0:2/3*\x) {$22$};
% \node (3) at (60:2/3*\x) {$33$};
% \node (4) at (120:2/3*\x) {$44$};
% \node (5) at (180:2/3*\x) {$55$};
% \node (6) at (240:2/3*\x) {$66$};
% \node (1) at (300:2/3*\x) {$11$};

% \draw (A.center)--(B.center)--(C.center) -- cycle;
% \draw (A.center)--(E.center);
% \draw (B.center)--(F.center);
% \draw (C.center)--(D.center);

% \end{tikzpicture}}

\subsection{Results for tree isomorphism}

We generated $100$ pairs of trees $(T_1,T_2)$ of size $n=10^i$ for each $i\in [\![1,6]\!]$, and for each case $T_1\simeq T_2$ or $T_1\not\simeq T_2$, following the procedure below:
\begin{itemize}
    \item for $T_1\simeq T_2$, we generate a random recursive tree \cite{zhang2015number} $T_1$ of size $n$, and generate $T_2$ as a copy of $T_1$;
    \item for $T_1\not\simeq T_2$, $T_1$ and $T_2$ are both random recursive trees of size $n$.
\end{itemize}

When processing a couple, we executed all five algorithms (original AHU, AHU with pigeonhole sort, AHU with timsort, AHU with primes and AHU with pregenerated primes) on that same couple, so that the results are comparable. The computation times we got are depicted in Figure~\ref{fig:isom_results}.

\begin{figure}[p]
\thisfloatpagestyle{empty}
    \centering
\begin{subfigure}[b]{\textwidth}
    \centering
\caption{$T_1\simeq T_2$}
\includegraphics[width=0.95\textwidth]{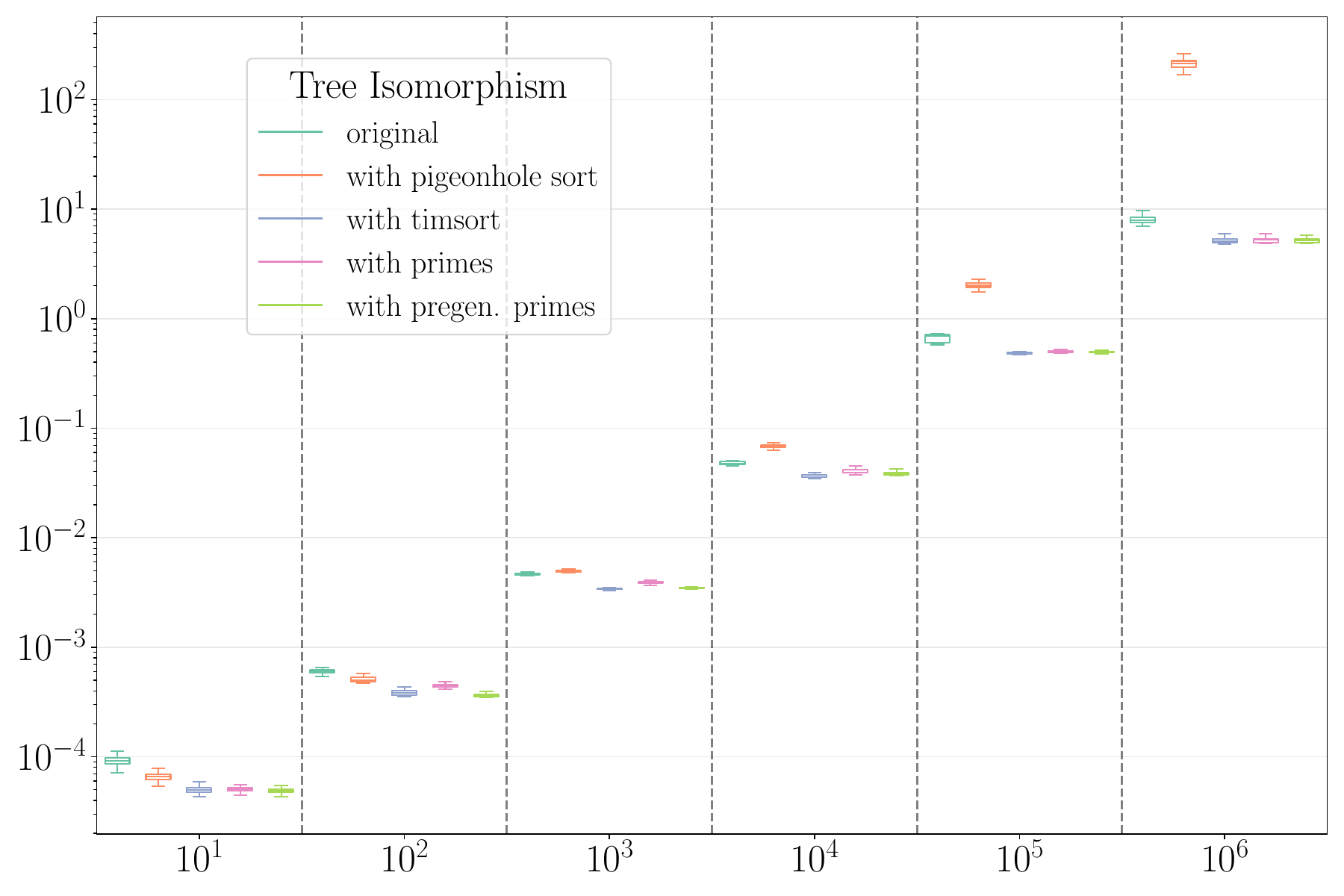}
\end{subfigure}

\begin{subfigure}[b]{\textwidth}
    \centering
\caption{$T_1\not\simeq T_2$}
\includegraphics[width=0.95\textwidth]{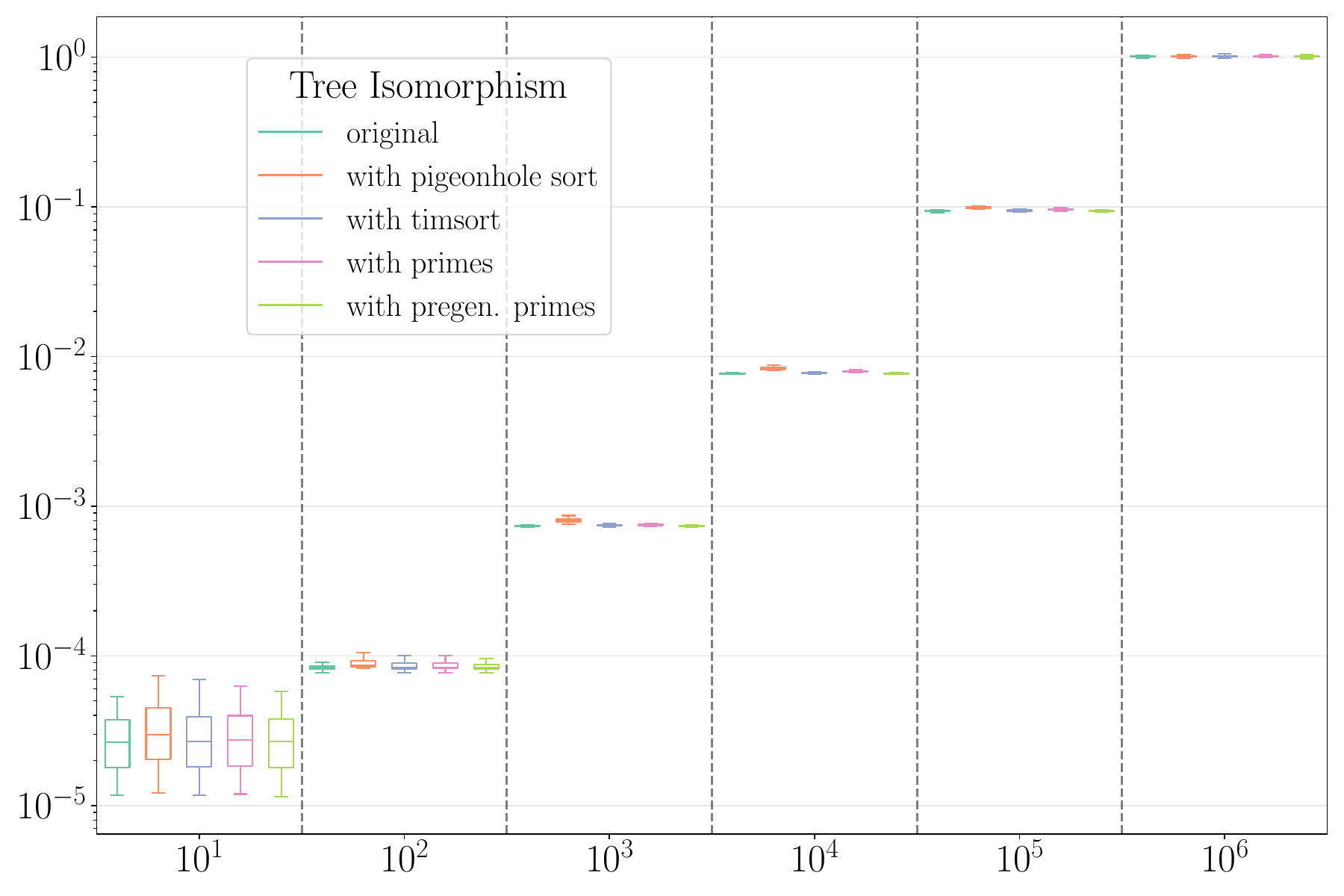}
\end{subfigure}
\caption{Computation times (in log scale) for tree isomorphism using different algorithms, according to the size of the trees, with 100 replicates for each size.}
    \label{fig:isom_results}
\end{figure}

As one might expect, in the case $T_1\simeq T_2$, the behaviour of AHU with pigeonhole is supralinear. However, this is not the case for the other variants, which are in fact faster than the original algorithm. This can be explained by the fact that AHU has a large constant, scanning each list several times during its execution. Furthermore, there doesn't seem to be any significant difference between AHU with timsort, with primes or with pregenerated primes (for the same tree, we found that timsort is actually slightly faster in most of the cases). It's not surprising that the two versions with prime numbers are similar: it has already been established that the complexity of generating the prime numbers is negligible compared with the other steps in the algorithm. However, the proximity between timsort on the one hand, and prime numbers on the other, suggests that at this scale it is just as fast to sort lists as it is to multiply them.

Whenever $T_1\not\simeq T_2$, all algorithms are roughly equally fast. This is likely due to the early stopping conditions: the distribution of random recursive trees generates trees that are too dissimilar for several levels to be visited during the execution.

\begin{remark}
 One could argue that $10^6$ is not very large, as an upper limit to our simulations. However, most of the real tree databases we are aware of already do not have trees of this size. Among the 8 datasets studied in \cite{azais2020weight}, the largest trees have a few thousand nodes; among the 5 studied in \cite{shin2014comprehensive}, a few hundred. If gigantic tree databases were to be built (for example, spanning trees from graphs with billions of nodes), it seems reasonable to imagine that they would be processed, in any case, with algorithms implemented in C, C++ or Rust rather than in Python (at the very least to be able to compute them efficiently, with the spanning trees example).
\end{remark}

\subsection{Results for DAG compression}

Here, we generated $100$ random recursive trees $T$ of size $n=10^i$ for each $i\in [\![1,6]\!]$. The results are provided in Figure~\ref{fig:dag_compression_results}. One can see that the pigeonhole variant seems to confirm a quadratic behaviour, while the timsort variant confirms its slight advantage over the versions with prime numbers. This time, a difference can be seen between the version that uses pregenerated primes and the one that does not (as might be expected given the theoretical complexity).

\begin{figure}[h!]

    \centering

\includegraphics[width=0.95\textwidth]{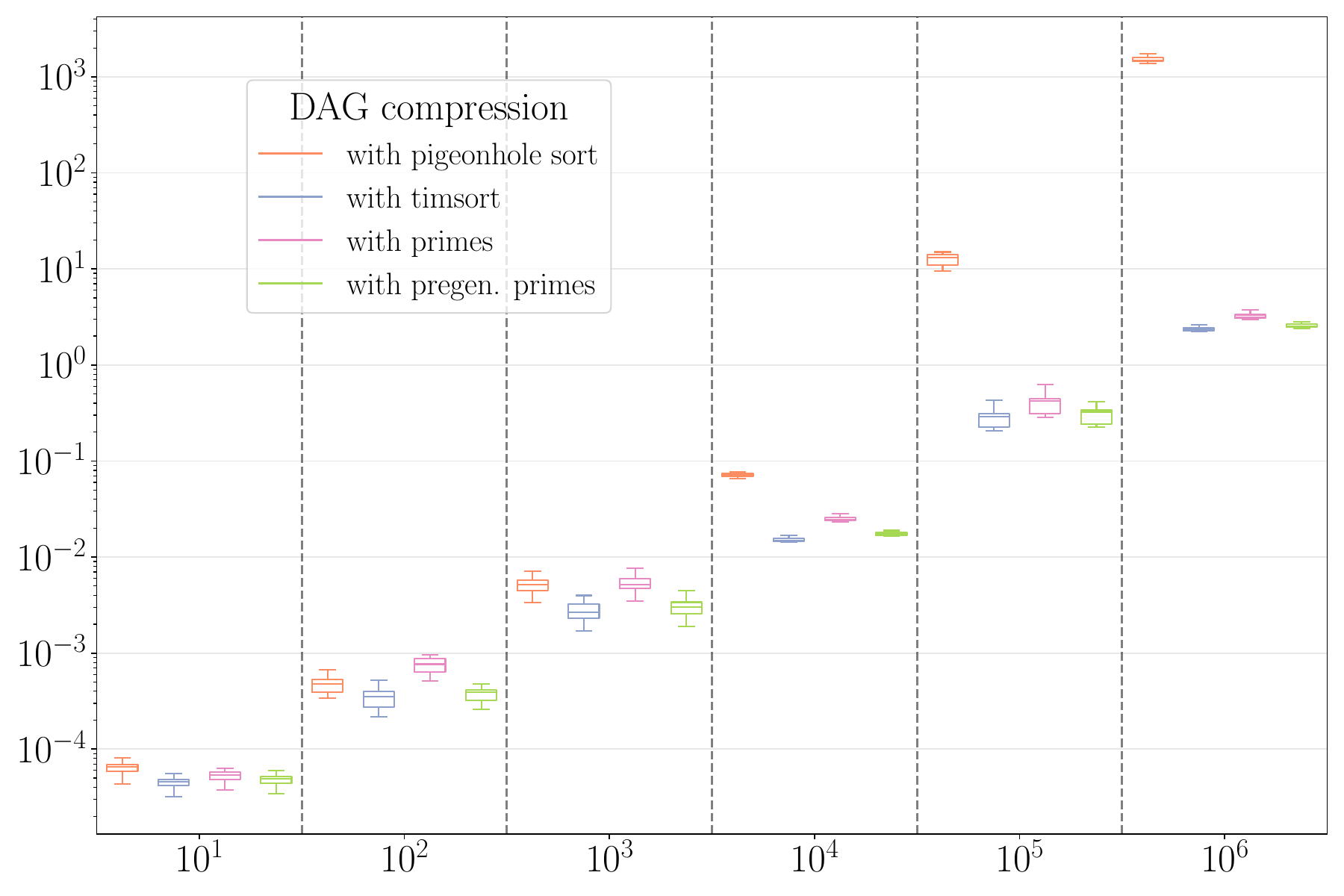}

\caption{Computation times (in log scale) for DAG compression using different algorithms, according to the size of the trees, with 100 replicates for each size.}
    \label{fig:dag_compression_results}
\end{figure}

\subsection*{Concluding remarks}

In this article, we have provided a new intuition for understanding the AHU algorithm, as well as several algorithmic variants that are straightforward to implement, albeit at the cost of increased complexity. However, we have shown that on trees of reasonable size, with a Python implementation, some of these variants were faster than the original algorithm. We have shown that a simple adaptation of our algorithms can also be used to calculate the DAG compression of trees.

Perhaps counter-intuitively, we have shown that in this context we can multiply lists of primes almost as quickly as we can sort lists of integers. However, if one had to pick just one variant, the one with timsort would probably be the simplest to implement and the most effective in practice -- bot for tree isomorphism and DAG compression. The version using prime numbers could possibly become competitive with timsort if the list multiplication and prime number generation operations were implemented in CPython (as is the case for timsort), but this is beyond the scope of this article.

One may also wonder whether our variant using prime numbers could be applied to other algorithms similar to AHU, such as (1-dimensional) Weisfeiler-Lehman algorithm for graph isomorphism. While this issue is outside the scope of this paper, and remains to be investigated, let us nonetheless mention two points that may prove challenging. First, the way Weisfeiler-Lehman operates can lead to processing as many colours as there are nodes in the graph, and therefore having to generate as many prime numbers -- which relates to the DAG compression case studied in this paper, for which our variant performed slightly worse than the others. Next, we would multiply lists whose size depends on the degree of the current node; in a dense or complete graph, this means lists whose size is comparable to the number of nodes in the graph. The complexity of performing these multiplications could prove far more expensive than for trees. Since Weisfeiler-Lehman can be implemented in $O((\#V+\#E) \log \#V)$ for a graph $G=(V,E)$, it remains to be investigated to which extent the additional complexities mentioned above exceeds that of the original algorithm. See \cite[Section~3.1]{kiefer2020power} and references therein for a more precise description of the Weisfeiler-Lehman algorithms.

\subsection*{Acknowledgements}
The author would like to thank Dr. Romain Azaïs and Dr. Jean Dupuy for their helpful suggestions on the first drafts of the article; as well as the anonymous reviewers who provided feedback on a previous version of this article.

\printbibliography

% \newpage
\appendix

\section{Proof of Lemma~\ref{width_bound_pigeonhole} and Lemma~\ref{width_bound}}\label{app:proof}

We conduct the proof of Lemma~\ref{width_bound_pigeonhole} and Lemma~\ref{width_bound} by first introducing a special tree, which for a given width has the smallest possible number of nodes, before observing how these two quantities are related.

\subsection{A special tree}\label{ss:special_tree}

\begin{minipage}[c]{0.5\textwidth}
Let $k\geq 1$ be a fixed integer. A tree $T$ such that $\width(T)=k$ can be obtained by placing $k$ trees $T_i$, $i\in[\![1,k]\!]$, under a common root, so that $T_i\not\simeq T_j$ for $i\neq j$. Note that this construction by no means encompasses all types of trees $T$ with $\width(T)=k$. On the other hand, by cleverly choosing the $T_i$'s, we can build a tree with the minimum number of nodes among all trees verifying $\width(T)=k$.

\medskip 

First, $T_1$ would be the tree with a unique node. Then, $T_2$ the only tree with two nodes. Then, $T_3$ and $T_4$ would be the two non-isomorphic trees with three nodes; $T_5$ to $T_8$ the four non-isomorphic trees with four nodes, and so on until we reach $T_k$. See Figure~\ref{fig:tree_ratio_max} for an example with $k=5$. It should be clear that this construction ensures that $\width(T)=k$ and $\#T$ is minimal. %The explicit generation of $T_1,\dots,T_k$ can be carried following \cite{nakano2003efficient}.
\end{minipage}\hfill
\begin{minipage}[c]{0.45\textwidth}
% \begin{figure}[H]
   \centering
\def\xscale{0.7}
\def\yscale{1}
\def\nodescale{0.7}
\begin{tikzpicture}[xscale=\xscale,yscale=\yscale]
\tikzstyle{fleche}=[-,>=latex]
\tikzstyle{noeud}=[draw,circle,fill=black,scale=\nodescale]

\node[noeud] (u) at (3,1) {};

\node[noeud] (u0) at (-0.5,0) {};

\node[noeud] (u1) at (1,0){};
\node[noeud] (u10) at (1,-1){};
\draw[fleche] (u1)--(u10);

\node[noeud] (u2) at (2.5,0){};
\node[noeud] (u20) at (2,-1){};
\node[noeud] (u21) at (3,-1){};
\draw[fleche] (u2)--(u20);
\draw[fleche] (u2)--(u21);

\node[noeud] (u3) at (4,0){};
\node[noeud] (u31) at (4,-1){};
\node[noeud] (u32) at (4,-2){};
\draw[fleche] (u3)--(u31);
\draw[fleche] (u31)--(u32);

\node[noeud] (u4) at (5.5,0){};
\node[noeud] (u41) at (5.5,-1){};
\node[noeud] (u42) at (5,-1){};
\node[noeud] (u43) at (6,-1){};
\draw[fleche] (u4)--(u41);
\draw[fleche] (u4)--(u42);
\draw[fleche] (u4)--(u43);

\draw[fleche] (u)--(u0);
\draw[fleche] (u)--(u1);
\draw[fleche] (u)--(u2);
\draw[fleche] (u)--(u3);
\draw[fleche] (u)--(u4);

\draw [decorate,
	decoration = {brace,mirror}] (-1,-2.5) --  (0,-2.5)node[pos=0.5,below=5pt] {$a_1$};

 \draw [decorate,
	decoration = {brace,mirror}] (0.5,-2.5) --  (1.5,-2.5)node[pos=0.5,below=5pt] {$a_2$};

  \draw [decorate,
	decoration = {brace,mirror}] (1.75,-2.5) --  (4.5,-2.5)node[pos=0.5,below=5pt] {$a_3$};

  \draw [decorate,
	decoration = {brace,mirror}] (4.75,-2.5) --  (6.5,-2.5)node[pos=0.5,below=5pt] {$k-b_3$};

\draw [decorate,
	decoration = {brace,mirror}] (-1,-3.25) --  (1.5,-3.25)node[pos=0.5,below=5pt] {$b_2$};
 
\draw [decorate,
	decoration = {brace,mirror}] (-1,-3.375) --  (4.5,-3.375)node[pos=0.5,below=5pt] {$b_3$};

\end{tikzpicture}
\captionof{figure}{The smallest tree so that $\width(T)=5$. We have $b_3<5\leq b_4$ and $\#T= 1 + 1\cdot a_1 + 2\cdot a_2 + 3\cdot a_3 + 4\cdot (5-b_3)=14$.}
    \label{fig:tree_ratio_max}
% \end{figure}
\end{minipage}

Following this construction, the total number of nodes in $T$, that we denote by $t_k$, is therefore closely related to the number of non-isomorphic trees and their cumulative sum. Let us denote $a_n$ the number of non-isomorphic trees of size $n$, and $b_n$ the number of non-isomorphic trees of size at most $n$ -- so that $b_n=\sum_{i=1}^n a_i$. Let $n$ be the integer so that $b_n < k \leq b_{n+1}$. All trees with size up to $n-1$ are used in the construction, as well as $k-b_n$ trees of size $n+1$ (no matter which ones).

\begin{minipage}[c]{0.5\textwidth}
 Therefore,
$$t_k = 1 + \sum_{i=1}^n i\cdot a_i + (n+1) (k-b_n).$$ Table~\ref{tab:trees} provides the first values for $a_n$, $b_n$ and $t_k$. Following the previous discussion, we have the following result.
\begin{lemma}\label{lemma:bound}
For any tree $T$, $\#T\geq t_{\width(T)}$.
\end{lemma}

\end{minipage}\hfill
\begin{minipage}[c]{0.45\textwidth}
    % \begin{table}[H]
    \centering
\begin{tabular}{c|cccccccc}
   $n$  & $1$ & $2$ & $3$ & $4$ & $5$ & $6$ & $7$ & $8$  \\
   \hline
    $a_n$ &  $1$& $1$& $2$& $4$& $9$& $20$& $48$&$ 115$ \\
    $b_n$ & $1$& $2$& $4$& $8$& $17$& $37$& $85$& $200$ \\
    \hline
    $t_n$ & $2$ & $4$& $7$&$10$ & $14$&$18$ &$22$ &$26$  \\
\end{tabular}
\captionof{table}{First values of $a_n$, $b_n$ and $t_n$. $a_n$ is sequence \href{https://oeis.org/A000081}{A000081} in OEIS, and $b_n$ sequence \href{https://oeis.org/A087803}{A087803}. See OEIS Foundation Inc. (2023), The On-Line Encyclopedia of Integer Sequences, Published electronically at \url{https://oeis.org}.}
    \label{tab:trees}
    % \end{table}
\end{minipage}

\subsection{Relationships between $k$ and $t_k$}

We require some preliminary results. We begin with the following lemma.

\begin{lemma}\label{lemma:technical}
Let $(u_n)_{n\in\mathbb{n}}$ be a sequence so that $u_n\underset{+\infty}{\sim} c \cdot d^n \cdot n^{-\alpha}$, with $c,\alpha\geq 0$ and $d> 1$. Then $\sum_{k=1}^n u_k \underset{+\infty}{\sim} c \cdot\frac{d}{d-1}\cdot d^n\cdot n^{-\alpha}$.
\end{lemma}
\begin{proof}
Obviously the sequence $\sum  d^n \cdot n^{-\alpha}$ diverges, and therefore we have $\sum_{k=1}^n u_k \underset{+\infty}{\sim} c \sum_{k=1}^n  d^k \cdot k^{-\alpha}$. Then,
$$\frac{c \sum_{k=1}^n  d^k \cdot k^{-\alpha}}{c \cdot \frac{d}{d-1}\cdot d^n\cdot n^{-\alpha}} = \frac{d-1}{d}\sum_{k=1}^n d^{k-n} \left(\frac{k}{n}\right)^{-\alpha} \underset{j=n-k}{=} \frac{d-1}{d}\sum_{j=1}^n \left(1-\frac{j}{n}\right)^{-\alpha} d^{-j}.$$

With bounds $1\leq \left(1-\frac{j}{n}\right)^{-\alpha}\leq \left(1-\frac{1}{n}\right)^{-\alpha}$, it is easy to see that the right-hand term goes to $1$ as $n\to\infty$.
% Taking the series expansion of $\left(1-\frac{j}{n}\right)^{-\alpha}$ in $0$, we have 
% $$(\star) = \frac{d}{d-1} \sum_{j=1}^n \left(1-\frac{\alpha j}{n}+\frac{\alpha(\alpha+1)}{2}\frac{j^2}{n^2}+\dots\right) d^{-j}= \frac{d}{d-1} \left(1+\frac{d}{n}\frac{\partial}{\partial d}\right)^{-\alpha}\sum_{j=1}^n d^{-j};$$
% noticing that $-j d^{-j}= \left(d\displaystyle\frac{\partial}{\partial d}\right)d^{-j}$, $j^2d^{-j}=\left(d\displaystyle\frac{\partial}{\partial d}\right)^2d^{-j}$ and so on. Adding negligible terms, we sum to infinity and obtain 
% \begin{align*}
%     (\star)&\underset{+\infty}{\sim} \frac{d-1}{d} \left(1+\frac{d}{n}\frac{\partial}{\partial d}\right)^{-\alpha}\displaystyle\sum_{j=1}^\infty d^{-j} = \frac{d-1}{d}\left(1+\frac{d}{n}\frac{\partial}{\partial d}\right)^{-\alpha} \frac{d}{d-1}\\
%     &= \frac{d-1}{d} \left(1-\frac{d\alpha}{n}\frac{\partial}{\partial d} + \frac{d^2\alpha(\alpha+1)}{2n^2}\frac{\partial^2}{\partial d^2} +\dots\right)\frac{d}{d-1}= 1+\frac{\alpha}{n(d-1)}+O\left(\frac{1}{n^2}\right)\\
%     &\xrightarrow[n \to \infty]{} 1.\qedhere
% \end{align*}
\end{proof}

From \cite[Section~2.3.4.4]{knuth1997art}, we have $a_n \underset{+\infty}{\sim} c\cdot d^n\cdot n^{-3/2}$ with $c\approx 0.439924$ and $d\approx 2.955765$. From Lemma~\ref{lemma:technical}, we immediately derive $b_n \underset{+\infty}{\sim} c\cdot \frac{d}{d-1} \cdot d^n \cdot n^{-3/2}$. Finally, noticing that $i\cdot a_i \underset{+\infty}{\sim} c\cdot d^i\cdot i^{-1/2}$, we derive from Lemma~\ref{lemma:technical} that $\sum_{i=1}^n i\cdot a_i \underset{+\infty}{\sim} c \cdot \frac{d}{d-1} \cdot d^n \cdot n^{-1/2}$.

We now derive our main results.

\begin{lemma}
$k n\underset{+\infty}{\sim}t_k$ -- with $n$ as defined in Subsection~\ref{ss:special_tree}, i.e. so that $b_n<k\leq b_{n+1}$.
\end{lemma}
\begin{proof}
We have
$$\frac{t_k}{kn}=\frac{1+\sum_{i=1}^n i\cdot a_i + (n+1)(k-b_n)}{kn} = \frac{1+(n+1)k}{nk} + \frac{\sum_{i=1}^n i\cdot a_i - (n+1)b_n}{nk}.$$

First, the left-hand term tends to $1$ as $k\to \infty$. We now prove that the right-hand term tends to $0$ as $k\to\infty$. Since $b_n<k$, we have
$$\left|\frac{\sum_{i=1}^n i\cdot a_i - (n+1)b_n}{nk}\right| < \left|\frac{\sum_{i=1}^n i\cdot a_i - nb_n}{nb_n}\right|+\frac{b_n}{nb_n}.$$
Notice that $\sum_{i=1}^ni\cdot a_i \sim nb_n$. By definition of $\sim$ and $o(\cdot)$ notations, for any (positive) sequences $u_n$ and $v_n$, $u_n\sim v_n \iff u_n-v_n = o(v_n) \iff \displaystyle\frac{u_n-v_n}{v_n}\to 0$, hence the result.
\end{proof}

\begin{lemma}
$\ln k \underset{+\infty}{\sim} \ln t_k \underset{+\infty}{\sim} n \ln d$ with $d\approx 2.955765$.
\end{lemma}
\begin{proof}
By definition, $b_n<k\leq b_{n+1}$ and $\displaystyle\sum_{i=1}^n i \cdot a_i< t_k\leq 1+\sum_{i=1}^{n+1} i\cdot a_i$. Hence, $k\sim b_n$ and $t_k\sim \displaystyle\sum_{i=1}^n i\cdot a_i$.

Taking the logarithm of the asymptotic equivalents provided earlier on both equations yields the result.
\end{proof}

Combining the two previous lemmas, we derive the following two corollaries.

\begin{corollary}
$k=O\left(\displaystyle\frac{t_k}{\ln(t_k)}\right)$
\end{corollary}

From Lemma~\ref{lemma:bound}, and since $x\mapsto \displaystyle\frac{x}{\log x}$ is increasing (for $x\geq 3$), we get that for any tree $T$, $\width(T) = O\left(\displaystyle\frac{\#T}{\log\#T}\right)$ -- hence Lemma~\ref{width_bound_pigeonhole} holds.

\begin{corollary}
$k\ln k=O(t_k)$.
\end{corollary}

From Lemma~\ref{lemma:bound}, we get that for any tree $T$, $\width(T)\ln\width(T) = O(\#T)$, which ends the proof of Lemma~\ref{width_bound}.

\section{Proof of Theorem~\ref{th:complexity_primes}}\label{app:proof_complexity_primes}

Following Proposition~\ref{prop:primes_generation}, the generation of primes is done in $O(\#T\log\log\#T)$. This term vanishes in the final complexity due to the upcoming term in $\#T\log \#T$. Let us denote by $p_{\width(T)}$ the largest prime needed by the algorithm. Fix $d\in[\![0,\depth(T)]\!]$ and $u\in T^d$. 

\paragraph{Complexity of multiplication} The complexity for multiplying two $n$-bits numbers varies from $O(n^2)$ for usual schoolbook algorithm, to $O(n\log n)$ \cite{harvey2021integer} -- even if this result is primarily theoretical, by the authors' own admission. Karatsuba algorithm, which is widely used, runs in $O(n^{\log 3})$ \cite{karatsuba1962multiplication}. This algorithm is actually used in Python when the numbers get large, and schoolbook otherwise. Let us denote the complexity of multiplication as $n\cdot M(n)$, with $M(n)$ varying from $\log n$ to $n$ depending on the algorithm used.

Multiplying two $n$-bits numbers together yields a $2n$-bits number. To compute the product of $m$ numbers of $n$ bits, we adopt a divide and conquer approach and multiply two numbers which themselves are the recursive product of $m/2$ numbers. This strategy leads to a complexity of $O(mn\cdot M(mn))$ by virtue of the Master Theorem \cite{bentley1980general}. Since computing $N(u)$ implies multiplying $\deg(u)$ primes with at most $\log p_{\width(T)}$ bits, this lead to a complexity of $O\left(\deg(u)\log p_{\width(T)} \cdot M\left(\deg(u) \log p_{\width(T)}\right)\right)$. Using \eqref{primes_bound} we have $p_{\width(T)}< \width(T)\left( \ln\width(T) + \ln\ln \width(T)\right)$. With Lemma~\ref{width_bound}, we have $p_{\width(T)}=O(\#T)$; thus a final complexity of $O(\deg(u) \log \#T \cdot M\left(\deg(u) \log\#T\right))$ for computing $N(u)$.

\paragraph{Other points} Testing whether or not $f(N(u))$ is defined in line~12 can be done in $O(1)$ since $N(u)$ is an integer, as per our assumption of perfect hash tables working with integers, strings and tuples. 

For comparing the multisets in line~16, we resort to pigeonhole sort as for Algorithm~\ref{algo:ideal_ahu}. Classic pigeonhole would have complexity  $O(\#T^d + p_n)$, where $p_n$ is the biggest prime in the list; but many holes will be unnecessary (as $c(u)$ is necessarily prime). Using a perfect hash table that associate to the $i$-th prime number the integer $i$, one can use only $n$ holes, one for each prime number, which reduces the complexity to $O(\#T^d + n)$. Since the primes are reallocated at each level, at level $d$ we need as many primes as there are different equivalence classes at that level -- i.e. $\#\lbrace c(u): u\in T^d\rbrace$. This number is $\leq \#T^d$, therefore the complexity of the sort collapses to $O(\#T^d)$.

\paragraph{Conclusion} Processing level $d$ thus requires $$O\left(\sum_{u\in T^d} \big[\deg(u) \cdot \log \#T \cdot M\left(\deg(u) \cdot \log \#T\right)\big] + \#T^d\right).$$

First, notice that $\sum_{u\in T^d} \deg(u) = \# T^{d-1}$. Bounding other occurrences of $\deg(u)$ by $\deg(T)$ and summing over $d$ leads to the claim.

\end{document}